\documentclass[runningheads]{llncs}
\usepackage{graphicx}
\usepackage{algorithm}
\usepackage{algpseudocode}
\usepackage{amssymb}
\usepackage{textcomp}
\usepackage{amsmath}

\newcommand{\V}{{\cal V}}
\newcommand{\cS}{{\cal S}}
\newcommand{\Prb}[1]{\mathbf{Pr} [#1]}
\newcommand{\Exp}[1]{\mathbf{E} [#1]}
\newcommand{\ceil}[1]{\lceil #1 \rceil}
\newcommand{\Endproof}{\hfill$\Box$\\}
\newcommand{\ket}[1]{\left| #1 \right\rangle}

\DeclareMathOperator{\RADIX}{RADIX}
\DeclareMathOperator{\SIGNUM}{SGN}
\DeclareMathOperator{\FST}{FST}
\DeclareMathOperator{\SG}{SG}
\DeclareMathOperator{\ADV}{ADV}
\DeclareMathOperator{\QC}{Q}
\DeclareMathOperator{\sgn}{sgn}

\begin{document}
\title{Noisy Tree Data Structures and Quantum Applications
}
%
%
\author{Kamil Khadiev\inst{1}\orcidID{0000-0002-5151-9908} \and
Nikita Savelyev\inst{1} \and Mansur Ziatdinov\inst{1}\orcidID{0000-0001-7415-2726}\and Denis Melnikov\inst{1}}
\authorrunning{K. Khadiev, N. Savelyev, M. Ziatdinov and D. Melnikov}
%
\institute{Institute of Computational Mathematics and Information Technologies, Kazan Federal University, Kazan, Russia\\
\email{kamilhadi@gmail.com}}
\maketitle              
\begin{abstract}
The paper presents a technique for constructing noisy data structures called a walking tree. We apply it for a Red-Black tree (an implementation of a Self-Balanced Binary Search Tree) and a segment tree. We obtain the same complexity of the main operations for these data structures as in the case without noise (asymptotically). We present several applications of the data structures for quantum algorithms.
Finally, we suggest new quantum solution for strings sorting problem and show the lower bound. The upper and lower bounds are the same up to a log factor. At the same time, it is more effective than classical counterparts. 

\end{abstract}
\section{Introduction}
Tree data structures are well-known and used in different algorithms. At the same time, when we construct algorithms with random behavior like randomized and quantum algorithms, we should consider error probability. We suggest a general method for updating a tree data structure in the noisy case, and we call the method {\em Walking tree}.
For a tree of height $h$, we consider an operation for processing all nodes from a root to a target node. Assume that the running time of the operation is $O(h\cdot T)$, where $T$ is processing-a-node running time. Then, in the noisy case, our technique allows us to do it with $O(\log (1/\varepsilon)+h\cdot T)$ running time,  where $\varepsilon$ is the error probability for the whole operation. Here we assume that navigation by the tree can have error probability. Note that the standard way to handle probabilistic navigation is the success probability boosting  (repetition of the noisy action) with $O(h\cdot\log (1/\varepsilon)+h\cdot T)$ complexity.

Our technique is based on results for the noisy Binary search algorithm from \cite{frpu94}. The authors of that paper present an idea based on the random walk algorithm for a balanced binary tree that can be constructed for binary search algorithm. We generalize the idea for a tree with any structure that allows us to apply the method to a big class of tree data structures. Two examples of such data structures are presented in the next paragraph. Note that different algorithms for noisy search especially, a noisy tree and graph processing and search were considered in \cite{p1989,kk2007,eks2016,dlu2021,dms2019,bkr2016,dkuz2017}.

We apply the technique to two tree data structures. The first one is Red-Black tree \cite{cormen2001} which is an implementation of self-balanced binary search tree \cite{cormen2001}. If the key comparing procedure has an error probability at most $p<0.5-\delta$ for some $0<\delta<0.5$, then our noisy self-balanced binary search tree allows us to do add, remove, and search operations in $O(\log (N/\varepsilon))$ running time. Here $\varepsilon$ is the error probability for a whole operation and $N$ is the number of nodes in the tree. In the case of $\varepsilon=1/Poly(N)$, we have $O(\log N)$ running time. So, in that case, the noisy key comparing procedure does not affect running time (asymptotically). At the same time, if we use the success probability boosting technique, then the running time is $O((\log N)^2)$.
The second one is the Segment tree \cite{momm2008,l2017guide}. If the indexes comparing procedure has an error probability at most $p<0.5-\delta$ for some $0<\delta<0.5$, then our noisy segment tree allows us to do update and request operations in $O(\log (N/\varepsilon))$ running time, where $\varepsilon$ is the error probability for a whole operation and $N$ is the number of leaves. In the case of $\varepsilon=1/Poly(N)$, we have $O(\log N)$ running time. So, we obtain similar advantage.

We use these data structures in the context of quantum computation. Quantum computation \cite{a2017,nc2010,aazksw2019part1,k2022lecturenotes,dw2001,quantumzoo} is one of hot topics in the last decades. 
Quantum algorithms have randomized behavior, so it is important to use noisy data structures for this model. 
 We use the quantum query model \cite{a2017,aazksw2019part1,k2022lecturenotes} as the main computational model for the quantum algorithms. We use the Walking tree method to solve the following problems.
 
 The first problem is the String Sorting problem. Assume that we have $n$ strings of length $l$. We want to sort them in lexicographical order. It is known \cite{hns2001,hns2002} that no quantum algorithm can sort arbitrary comparable objects faster than $O(n\log n)$. At the same time, several researchers tried to improve the hidden constant \cite{oeaa2013,oa2016}. Other researchers investigated the space-bounded case \cite{k2003sort}. We focus on sorting strings. In the classical case, we can use an algorithm that is better than arbitrary comparable objects sorting algorithms. It is Radix sort that has $O(nl)$ query complexity \cite{cormen2001} for a finite-size alphabet. It is also a lower bound for classical (randomized or deterministic) algorithms that is $\Omega(nl)$. 
There is a quantum algorithm \cite{ki2019,kiv2022} for the string sorting problem that has query complexity $O(n(\log n)\cdot \sqrt{l})=\tilde{O}(n \sqrt{l})$, where $\tilde{O}$ does not consider log factors. We suggest a simpler implementation based on a noisy red-black tree. These algorithms show quantum speed-up. Additionally, we show the lower bound for quantum algorithms that is $\Omega(n\sqrt{l}/\sqrt{\log n})$. So, the quantum solution is optimal up to a log factor.
 
The second problem is the Auto-Complete Problem. We have two kinds of queries: adding a string $s$ to the dictionary $\cS$ and querying the most frequent complement of a string $t$ from the dictionary. We call $s$ a complement of $t$ if $t$ is a prefix of $s$. Assume that $L$ is the total sum of all lengths of strings from all queries.
%
We solve the problem using quantum string comparing algorithm 
\cite{bjdb2017,ki2019} 
and noisy Red-Black tree. The running time of the quantum algorithm is $O(\sqrt{nL}\log n)$. The lower bound for quantum running time is $\Omega(\sqrt{L})$. At the same time, the best classical algorithm based on trie(prefix tree)
\cite{knuth73} 
has $O(L)$ running time. That is also the classical (deterministic or randomize) lower bound $\Omega(L)$. So, we obtain quantum speed-up if most of the strings have $\Omega((\log n)^2)$ length.

The structure of this paper is as follows. Section \ref{sec:prelims} contains preliminaries. We present the main technique in Section \ref{sec:walking-tree}. Section \ref{sec:noisy-trees} contains a discussion of the noisy self-balanced binary search tree and the noisy segment tree. A discussion of String Sorting Problem is in Section \ref{sec:sort}, and a discussion of Auto-Complete Problem is in Section \ref{sec:appl}.

\section{Preliminaries}\label{sec:prelims}
In the paper, for two strings $s$ and $t$, the notation $s<t$ means $s$ precedes $t$ in the lexicographical order. Let $|s|$ be the length of a string $s$.

\textbf{Graph Theory.}
Let us consider a rooted tree $G$. Let $\V(G)$ be a set of nodes (vertices), and ${\cal E}(G)$ be a set of edges. Let one fixed node be the root of the tree. Assume that we can obtain it using a procedure $\textsc{GetTheRoot}(G)$.
A path $P$ is a sequence of nodes $(v_1,\dots,v_k)$ that are connected by edges, i.e. $(v_i,v_{i+1})\in E$ for all $i\in\{1,\dots,k-1\}$. Note, that there are no duplicates among $v_1,\dots,v_k$. Here, $k$ is the length of the path.  We use $v\in P$ notation if there is $j$ such that $v_j=v$. The notation is reasonable because there are no duplicates in a path. Note that for any two nodes $u$ and $v$ the path between them is unique because $G$ is a tree. 
The distance $dist(v,u)$ between two nodes $v$ and $u$ is the length of the path between them. A height of a node $v$ is the distance between it and the root that is $dist(root,v)$. Let $h(G)=\max_{v\in{\cal V}(G)}dist(root,v)$ be the tree's height which is the length of the path between the root and the farthest node.

For each node $v$ we can define a parent node $\textsc{Parent}(v)$, it is a node such that $dist(root,\textsc{Parent}(v))+1=dist(root,v)$ and $(\textsc{Parent}(v),v)\in {\cal E}(G)$. Additionally, we can define a set of children $\textsc{Children}(v)=\{u: \textsc{Parent}(u)=v\}$.

\textbf{Quantum Query Model.}
In Section \ref{sec:appl} we suggest quantum algorithms as applications for our data structures. We have only one quantum subroutine, and the rest part of the algorithm is classical.   
One of the most popular computation models for quantum algorithms is the query model.
We use the standard form of the quantum query model. 
Let $f:D\rightarrow \{0,1\},D\subseteq \{0,1\}^M$ be an $M$ variable function. We wish to compute on an input $x\in D$. We are given oracle access to the input $x$, i.e. it is implemented by a specific unitary transformation that is usually defined as $\ket{i}\ket{z}\ket{w}\rightarrow \ket{i}\ket{z+x_i\pmod{2}}\ket{w}$ where the $\ket{i}$ register indicates the index of the variable we are querying, $\ket{z}$ is the output register, and $\ket{w}$ is some auxiliary work-space. The operation is implemented by the CNOT gate. An algorithm in the query model consists of alternating applications of arbitrary unitaries independent of the input and the query unitary, and measurement in the end. The smallest number of queries for an algorithm that outputs $f(x)$ with probability $\geq \frac{2}{3}$ on all $x$ is called the quantum query complexity of the function $f$, and  $Q(f)$ notation is used for it. We use the running time term instead of query complexity for removing confusion with ``query'' in the definition of problems in Section \ref{sec:appl}. Note that in the general case, we can consider a function $f$ with non-Boolean arguments. It can be simulated by a Boolean-argument-function case using a binary representation of arguments.
We refer the readers to \cite{nc2010,a2017,aazksw2019part1,k2022lecturenotes} for more details on quantum computing. 

\section{Main Technique: A Walking Tree} \label{sec:walking-tree}
In this section, we present a rooted tree that we call {\em a walking tree}. It is a utility data structure for noisy computation for the main data structure. 
Here we use it for the following data structures:
%
(i) Binary Search Tree. We assume that elements comparing procedures can have errors.
 %
    (ii) Segment Tree. We assume that an indexes (borders of segments) comparing procedure can have errors.

Note that the walking tree is a general technique, and it can be used for other tree data structures.
Let us present the general idea of the tree. The technique is motivated by \cite{frpu94}.
Assume that we have a rooted tree $G$. We want to do an operation on the tree that is moving from the root to a specific ({\em target}) node of the tree.
Assume that we have the following procedures:
 %
 $\textsc{GetTheRoot}(G)$ returns the root node of the tree $G$.
 %
 $\textsc{SelectChild}(v)$ returns the child of the node $v\in \V (G)$ that should be reached from the node $v$. 
 %
 $\textsc{IsItTheTarget}(v)$ returns $True$ if the node $v\in \V (G)$ is the last node that should be visited in the operation; and returns $False$ otherwise. 
 %
 $\textsc{ProcessANode(v)}$ processes the node $v\in \V (G)$  in the required way.
 %
    $\textsc{IsANodeCorrect(v)}$  returns $True$ if the node $v\in \V (G)$ should be visited during the the operation; and returns $False$ if the node is visited because of an error.

Assume that the operation has the following form (Algorithm \ref{alg:maincomand}).

\vspace{-0.7cm}
\begin{algorithm}[ht]
    \caption{An operation on the tree $G$} \label{alg:maincomand}
    \begin{algorithmic}
        \State $v \gets \textsc{GetTheRoot}(G)$
        \State $\textsc{ProcessANode}(v)$
        \While{$\textsc{IsItTheTarget}(v)=False$}
        \State $v \gets\textsc{SelectChild}(v)$
         \State $\textsc{ProcessANode}(v)$
        \EndWhile
    \end{algorithmic}
\end{algorithm}

\vspace{-0.7cm}
Let us consider the operation such that ``navigation'' procedures (that are $\textsc{SelectChild}$, $\textsc{IsANodeCorrect}$, and $\textsc{IsItTheTarget}$) can return an answer with an error $p$, where $p<0.5-\delta$, where $0<\delta\leq 0.5$. We assume that error events are independent.
Our goal is to do the operation with an error $\varepsilon$. Note that in the general case, $\varepsilon$ can be non-constant and depend on the number of tree nodes.
Let $h=h(G)$ be the height of the tree. The standard technique is boosting success probability. On each step, we repeat $\textsc{SelectChild}$ procedure $O(\log (h/\varepsilon))$ times and choose the most frequent answer. In that case, the error probability of the operation is at most $\varepsilon$, and the running time of the operation is $O(h\log(h/\varepsilon)+h\cdot T)$, where $T$ is complexity of $\textsc{ProcessANode}$ procedure.
Our goal is to have $O(h+\log (h/\varepsilon)+h\cdot T) = O(\log (h/\varepsilon)+h\cdot T)$ running time.

Let us construct a rooted tree $W$ by $G$ such that the set of nodes of $W$ has a one-to-one correspondence to the nodes of $G$ and the same with sets of edges. We call $W$ a {\em walking tree}.
Let $\lambda_W:\V (W) \leftrightarrow \V (G)$ and $\lambda_G:\V (G) \leftrightarrow \V (W)$ be  bijections between these two sets. For simplicity, we define procedures for $W$ similar to the procedures for $G$. Suppose $u\in \V (W)$, then
$\textsc{GetTheRoot}(W)=\lambda_G(\textsc{GetTheRoot}(G));$ 
$\textsc{SelectChild}(u)=\lambda_G(\textsc{SelectChild}(\lambda_W(u)));$
$\textsc{IsItTheTarget}(u)=\textsc{IsItTheTarget}(\lambda_W(u));$  
$\textsc{IsANodeCorrect(u)}=\textsc{IsANodeCorrect}(\lambda_W(u)).$
Note that the navigation procedures are noisy (have an error). We reduce the error probability to $0.1$ by constant number of repetitions (using the boosting success probability technique).
Additionally, we associate a counter $c(u)$ with a vertex $u\in \V(W)$ that is a non-negative integer number. Initially, values of counters are $0$ for all nodes, i.e. $c(u)=0$ for each $u\in\V (W)$.

We invoke a random walk by the walking tree $W$. The walk starts from the root node $\textsc{GetTheRoot}(W)$. Let us discuss processing $u\in \V (W)$.
Firstly, we check the counter's value $c(u)$. If $c(u)=0$, then we do steps from 1.1 to 1.3. 
    
    {\bf Step 1.1.} We check the correctness of current node using $\textsc{IsANodeCorrect}(u)$ procedure.  If the result is $True$, then we go to Step 1.2. If the result is $False$, then we are here because of an error, and we go up by changing $u\gets \textsc{Parent}(u)$. If the node $u$ is the root, then we stay in $u$.
 
    {\bf Step 1.2.}  We check whether the current node is target using $\textsc{IsItTheTarget}(u)$ procedure. If it is $True$, then we increase the counter $c(u)\gets 1$. If it is  $False$, then we go to Step 1.3.

{\bf Step 1.3.} We go to the children $u\gets\textsc{SelectChild}(u)$. 

If $c(u)>0$, then we do Step 2.1. We can say that the counter $c(u)$ is a measure of confidence that $u$ is the target node. If $c(u)=0$, then we should continue walking. If $c(u)>0$, then we think that $u$ is the target node. If a bigger value of $c(u)$ means we are more confident that it is the target node.
    
    {\bf Step 2.1.} If $\textsc{IsItTheTarget}(u)=True$, then we increase the counter $c(u)\gets c(u)+1$.  Otherwise, we decrease the counter $c(u)\gets c(u)-1$. So, we become more or less confident in the fact that the node $u$ is the target. 

The walking process stops in $s$ steps, where $s=O(h+\log(1/\varepsilon))$. The stopping node $u$ is the target one. After that, we do the operation with the original tree $G$. We store path in $Path=(v^1,\dots,v^k)$, such that $v^k=\lambda_W(u)$, $v^i=\textsc{Parent}(v^{i+1})$, and $v_1$ is the root node of $G$. Then, we process them using $\textsc{ProcessANode}(v^i)$ for $i$ from $1$ to $k$.
One step of the walking process on the walking tree $W$ is presented in Algorithm \ref{alg:onestep} as a procedure $\textsc{OneStep}(u)$ that accepts the current node $u$ and returns the new node.  The whole algorithm is presented in Algorithm \ref{alg:walking}.

\begin{algorithm}[ht]
    \caption{One step of the walking process, $\textsc{OneStep}(u)$. The input is $u\in \V(W)$ and the result is the node for the next step of the walking.} \label{alg:onestep}
    \begin{algorithmic}
        \If{$c(u)=0$}
           \If{$\textsc{IsANodeCorrect}(u)=False$}\Comment{Step 1.1}
              \If{$u\neq\textsc{GetTheRoot}(W)$}
                 \State $u\gets \textsc{Parent}(u)$
              \EndIf
           \Else
              \If{$\textsc{IsItTheTarget}(u)=True$} \Comment{Step 1.2}
                 \State $c(u)\gets 1$
              \Else
                 \State $u\gets\textsc{SelectChild}(u)$ \Comment{Step 1.3}
               \EndIf
           \EndIf
        \Else
           \If{$\textsc{IsItTheTarget}(u)=True$}\Comment{Step 2.1}
              \State $c(u)\gets c(u)+1$
           \Else
              \State $c(u)\gets c(u)-1$
           \EndIf
        \EndIf
    \end{algorithmic}
\end{algorithm}


\begin{algorithm}[ht]
    \caption{The walking algorithm for $s=O(h+\log(1/\varepsilon))$ steps.} \label{alg:walking}
    \begin{algorithmic}
        \State $u\gets \textsc{GetTheRoot}(W)$
        \For{$j \in\{1,\dots s\}$}
            \State $u\gets \textsc{OneStep}(u)$
        \EndFor
        \State $v\gets \lambda_W(u)$, $Path = (v)$, $k=1$
        \While{$v\neq \textsc{GetTheRoot}(G)$}
            \State $v\gets \textsc{Parent}(v)$
            \State $Path=v \circ Path$\Comment{Here $\circ$ means the concatenation of two sequences. The line adds the node to the beginning of the path sequence}
            \State $k\gets k+1$\Comment{The length of the path sequence}
        \EndWhile
        \For{$i\in\{1,\dots k\}$}
            \State $\textsc{ProcessANode}(v^i)$ 
        \EndFor
    \end{algorithmic}
\end{algorithm}
Let us discuss the algorithm and its properties.
On each node, we have two options, we go in the direction of the target node or the opposite direction.

Assume that $c(u)=0$ of the current node $u$.
If we are in a wrong branch, then the only correct direction is the parent node. If we are in the correct branch, then the only correct direction is the correct child node. All other directions are wrong.
Assume that $c(u)>0$.
If we are in the target node, then the only correct direction is increasing the counter, and the wrong direction is decreasing the counter. Otherwise, the only correct direction is decreasing the counter.

Choosing the direction is based on the results of at most three invocations of navigation procedures ($\textsc{SelectChild}$, $\textsc{IsANodeCorrect}$, and $\textsc{IsItTheTarget}$). Remember that we reach $0.1$ error probability using a constant number of repetitions. Due independence of error events, the total error probability of choosing a direction is at most $0.3$.
So,the  probability of moving in correct direction is at least $2/3$ and for a wrong direction it is at most $1/3$.
Let us show that if $s=O(h+\log(1/\varepsilon))$, then a error probability for an operation on $G$ is $\varepsilon$ 

\vspace{-0.2cm}
\begin{theorem}\label{th:wt-compl}
Suppose, $s=O(h+\log(1/\varepsilon))$. Then, Algorithm \ref{alg:walking} does the same action as Algorithm \ref{alg:maincomand} with error probability $\varepsilon\in (0,0.5)$. (See Appendix \ref{apx:wt-compl}.)
\end{theorem}

In the next section, we show several applications of the technique.
\section{Noisy Tree Data Structures}\label{sec:noisy-trees}

\subsection{Noisy Binary Search Tree}
Let us consider a Self-Balanced Search Tree \cite{cormen2001}. It is a binary rooted tree $G$. Let $n=|\V(G)|$. We associate a comparable element $\alpha(v)$ with a node $v\in\V(G)$. (i) For a node $v\in\V(G)$, we have $\alpha(v')<\alpha(v)$ where $v'$ is from the left sub-tree  of $v$; and $\alpha(v'')>\alpha(v)$ where $v''$ is from the right sub-tree of $v$. 
    (ii) The height of the tree $h(G)=O(\log n)$.

As an implementation of {\em Self-Balanced} Search Tree, we use Red-Black Tree \cite{cormen2001,g78}.  It allows us to add and remove a node with a specific value with $O(\log n)$ running time. 
Assume that the comparing procedure of two elements has an error $p$.
Each operation (remove and add an element) has three steps: searching, doing the action (removing or adding), re-balancing. Re-balancing does not invoke comparing operations, that is why it does not have an error. So, the only ``noisy'' procedure (can have an error) is searching. Let us discuss it.

Let us associate $\beta(v)$ and $\gamma(v)$ with a node $v\in\V(G)$. That are left and right bounds for $\alpha(v)$ with respect to the ancestor nodes. Formally,
$\beta(v)<\min\{\alpha(v'):  v$ is an ancestor of $v'\}$, $\gamma(v)>\max\{\alpha(v'):  v$ is an ancestor of $v'\}$. We can compute them as follows. If $v$ is the root, then $\beta(v)=-\infty$,  and $\gamma(v)=+\infty$. Here $-\infty$ and $+\infty$ are constants that are apriori less and more than any $\alpha(v')$ for $v'\in \V(G)$. 
Let $v$ be a non-root node. If $v$ is the left child of $\textsc{Parent}(v)$, 
then $\beta(v)=\beta(\textsc{Parent}(v))$,
 $\gamma(v)=\alpha(\textsc{Parent}(v))$. 
  If $v$ is the right one, then $\beta(v)=\alpha(\textsc{Parent}(v))$, 
$\gamma(v)=\gamma(\textsc{Parent}(v))$.
Assume that a comparing function for elements $\textsc{Compare}(a,b)$ returns 
    $-1$ if $a<b$;
    $+1$ if $a>b$;
     $0$ if $a=b$.
An error probability for the function is $p<0.5-\delta$ for some $0<\delta<0.5$.
Let us present each of the required procedures for searching an object $x$ operation.
    $\textsc{GetTheRoot}(G)$ is for the root node of $G$.
  $\textsc{SelectChild}(v)$ returns the left child of $v$ if $\textsc{Compare}(\alpha(v),x)=-1$; and returns the right child if $\textsc{Compare}(\alpha(v),x)=+1$.
 $\textsc{IsItTheTarget}(v)$ returns $True$ iff $Compare(\alpha(v),x)=0$. 
$\textsc{ProcessANode(v)}$ do nothing.
  $\textsc{IsANodeCorrect(v)}$  returns $True$ iff $\beta(v)< x < \gamma(v)$, formally, $\textsc{Compare}(\beta(v),x)= -1$ and $\textsc{Compare}(x,\gamma(v))=-1$.
%
The presented operations satisfy all requirements. Let us present the complexity result that directly follows from Theorem \ref{th:wt-compl}.

\vspace{-0.2cm}
\begin{theorem}\label{th:bst-compl}
Suppose the comparing function for elements of Red-Black Tree has an error $p<0.5-\delta$ for some $0<\delta<0.5$. Then, using the walking tree, we can do searching, adding and removing operations with $O(\log(n/\varepsilon))$ running time and an error probability $\varepsilon$.
\end{theorem}

If $\varepsilon=1/Poly(n)$, then the ``noisy'' setting does not affect asymptotic complexity.
\begin{corollary}\label{cr:bst-compl}
Suppose the comparing function for elements of Red-Black Tree has an error $p<0.5-\delta$ for some $0<\delta<0.5$.
Then, using the walking tree, we can do searching, adding, and removing operations with $O(\log n)$ running time and an error probability $1/Poly(n)$.
\end{corollary}

\subsection{Noisy Segment Tree}\label{sec:st}
We consider a standard segment tree data structure \cite{momm2008,l2017guide} for an array $b=(b_1,\dots, b_n)$ for some integer $n>0$. 
The segment tree is a full binary tree such that each node corresponds to a segment of the array $b$. If a node $v$ corresponds to a segment $(b_{left},\dots, b_{right})$, then we store a value $\alpha(v)$ that represents some information about the segment. Let us consider a function $f$ such that $\alpha(v)=f(b_{left},\dots, b_{right})$. A segment of a node is the union of segments that correspond to its two children. Typically, the children correspond to segments $(b_{left},\dots, b_{mid})$ and $(b_{mid+1},\dots, b_{right})$, for $mid = \lfloor(left+right)/2\rfloor$. We consider  $\alpha(v)$ such that it can be computed by values of two children $\alpha(v_l)$ and $\alpha(v_r)$, where $v_l$ and $v_r$ are left and right children of $v$. Leaves correspond to single elements of the array $b$. As an example, we can consider integer values $b_i$ and sum $\alpha(v)=b_{left}+\dots+b_{right}$ as the value in a vertex $v$ and a corresponding segment $(b_{left},\dots, b_{right})$.
The data structure allows us to invoke the following requests in $O(\log n)$ running time.

    {\bf Update.} Parameters are an index $i$ and an element $x$ ($1\leq i\leq n$). The procedure assigns $b_i\gets x$. For this goal, we assign $x$ for the leaf $w$ that corresponds to the $b_i$ and update ancestors of $w$.
 
{\bf Request.}  Parameters are two indexes $i$ and $j$ ($1\leq i\leq j \leq  n$), the procedure computes $f(b_i,\dots,b_j)$.

The main part of both operations is the following. For the given root node and an index $i$, we should find the leaf node corresponding to $b_i$. The main step is the following. If we are in a node $v$ with associated segment $(b_{left},\dots, b_{right})$, then we compare $i$ with a middle element $mid = \lfloor(left+right)/2\rfloor$ and choose the left or the right child.
Assume that we have a comparing function for indexes $\textsc{Compare}(a,b)$ that returns 
$-1$ if $a<b$;
$+1$ if $a>b$;
 $0$ if $a=b$.
The comparing function returns the answer with an error $p<0.5-\delta$ for some $0<\delta<0.5$.

Let us present each of the required procedures for searching the leaf with index $i$ in a segment tree $G$.
%
$\textsc{GetTheRoot}(G)$ returns the root node of the segment tree $G$.
%
For $mid = \lfloor(left+right)/2\rfloor$, and the segment $(b_{left},\dots, b_{right})$ associated with a node $v$, the function $\textsc{SelectChild}(v)$ returns the left child of $v$ if $Compare(mid,i)\leq 0$; and returns the right child othewrwise.
%
$\textsc{IsItTheTarget}(v)$ returns $True$ if $left=i=right$, formally, $Compare(left,i)=0$ and $Compare(right,i)=0$; and returns $False$ otherwise. Here  the segment $(b_{left},\dots, b_{right})$ is associated with $v$.
%
$\textsc{ProcessANode}(v)$ recomputes $\alpha(v)=f(b_{left},\dots, b_{right})$ according to the values of $\alpha$ in the left and the right children.
%
$\textsc{IsANodeCorrect}(v)$  returns $True$ if $left\leq  i \leq right$, formally, $\textsc{Compare}(left,i)\leq 0$ and $\textsc{Compare}(x,right)\leq 0$; and returns $False$ otherwise. Here the segment $(b_{left},\dots, b_{right})$ is associated with $v$.
%
The presented operations satisfy all requirements. Let us present the complexity result that directly follows from Theorem \ref{th:wt-compl}.

\begin{theorem}\label{th:st-compl}
Suppose, the comparing function for indexes of a segment tree is noisy and has an error $p<0.5-\delta$ for some $0<\delta<0.5$. Then, using the walking tree, we can do update and request operations with $O(\log(n/\varepsilon))$ running time and an error probability $\varepsilon$.
\end{theorem}

If we take $\varepsilon=1/Poly(n)$, then the ``noisy'' setting does not affect asymptotic complexity.
\begin{corollary}\label{cr:st-compl}
Suppose, the comparing function for indexes of a segment tree is noisy and has an error $p<0.5-\delta$ for some $0<\delta<0.5$.
Then, using the walking tree, we can do update and request operations with $O(\log n)$ running time and an error probability $1/Poly(n)$
\end{corollary}
\subsection{Analysis, Discussion, Modifications}
There are different additional operations with a segment tree. One such example is the segment tree with range updates. In this modification, we can update the values $b_i$ for a range $i\in\{\ell,\dots,r\}$ by a value in one request. The reader can find more information in \cite{l2017guide} and examples of applications in \cite{kr2021a,kr2021b}. The main operation with a noisy comparing procedure is the same. So, we can still use the same idea for such modifications of the segment tree.

\begin{remark}\label{rm:st-memory}
If the segment tree is constructed for an array $b_1,\dots,b_n$, then we can extend it to $b_1,\dots,b_n,\dots, b_k$, where $k=2^{\lceil\log_2 n\rceil}$ that is closest to $n$ power of $2$ and $b_{n+1},\dots,b_{k}$ are neutral element for the function $f$.
If we have a vertex $v$ and two borders $left$ and $right$ of the segment associated with $v$, then we always can compute the segments for the left and the right children that are $b_{left},\dots,b_{mid}$ and $b_{mid+1},\dots,b_{right}$ for $mid=(left+right)/2$. Additionally, we can compute the segment for the parent that is $b_{left},\dots,b_{pright}$, where $pright=2\cdot right - left$ if the node is the left child of its parent. If the node is the right child of its parent, then the parent's segment is $b_{pleft},\dots,b_{right}$, where $pleft=2\cdot left - right$. Therefore, we should not store the borders of a segment in a node, and we can compute them during the walk on the segment tree. Additionally, we should not construct the walking tree. We can keep it in mind and walk by the segment tree itself using only three variables: the $left$ and $right$ borders of the current segment and a counter if required.
\end{remark}

If we have access to the full segment tree, including leaves, then we can do operations without the walking tree. We can use the noisy binary search algorithm \cite{frpu94} for searching the leaf that corresponds to the index $i$, and then process all the ancestors of the leaf.
There are at least two useful scenarios for a noisy segment tree.

1.
We have access only to the root and have no direct access to leaves.

2. The second one is the compressed segment tree. If initially, all elements of the array $b$ are empty or neutral for the function $f$, then we can compress a  subtree with one node with a label of a segment with empty elements. On each step, we do not store any subtree if it has only neutral elements. In that case, we store only a root of this tree and mark it as a subtree with neutral elements. It is reasonable if $n$ is very big and storing the whole tree is very expensive. In that case, we can replace the noisy binary tree with the noisy self-balanced search tree from the previous section. The search tree stores the updated elements in leaves, and we can search the required index in this data structure. At the same time, the noisy segment tree uses much less memory with respect to the Remark \ref{rm:st-memory}. That is why a noisy segment tree is more effective in this case too.


\section{Quantum Sort Algorithm for Strings}\label{sec:sort}
As one of the interesting applications, we suggest applications from quantum computing \cite{nc2010,a2017,aazksw2019part1,k2022lecturenotes}.
As objects with noisy comparing, we use strings. There is an algorithm for two strings $t$ and $s$ that compares them in $O\left(\sqrt{\min(|s|,|t|)}\right)$ running time 
\cite{kiv2022}
, where  $|s|$ and $|t|$ are lengths of $s$ and $t$ respectively. The algorithm is based on modifications 
\cite{k2014,kkmsy2022}
of Grover's search algorithm \cite{g96,bbht98}. The result is the following

\begin{lemma}[\cite{kiv2022}]\label{lm:strcmp}
There is a quantum algorithm that compares two strings $s$ and $t$ of lengths $|s|$ and $|t|$ in the lexicographical order with  $O(\sqrt{\min(|s|,|t|)}\log \xi^{-1})$ running time and error probability $O(\xi)$ for $0<\xi<1$.
\end{lemma}

Let us take $\xi=0.1$ and use the procedure as a string comparing procedure. We consider an application of the noisy search tree in this  and in the next sections. Let us discuss the Strings Sorting problem as one of the applications.
\paragraph{Problem.} There are $n$ strings $s^0,\dots,s^{n-1}$ of size $l$ for some positive integers $n$ and $l$. The problem is to find a permutation $(j_0,\dots,j_{n-1})$ such that $s^{j_i}< s^{j_{i+1}}$, or $s^{j_i}=s^{j_{i+1}}$ and $j_i<j_{i+1}$ for each $i\in\{0,\dots,n-2\}$. 

The Quantum sorting algorithm for strings was presented in \cite{ki2019,kiv2022}. The running time of the algorithm is $O(n\sqrt{l}\log n)$.
We can present the algorithm with the same complexity but in a simpler way. Assume that we have a noisy self-balanced binary search tree with strings as keys and quantum comparing procedure from Lemma \ref{lm:strcmp}. We assume that the comparing procedure compares indexes in a case of equal strings.
In fact, we store indexes of the strings in nodes.  We assume that if a node stores a key index $i$, then any node from the left subtree has a key index $j$ such that $s^j<s^i$ or ($s^j=s^i$ and $j<i$); and any node from the right subtree has a key index $j'$ such that $s^{j'}>s^i$ or ($s^j=s^i$ and $j'>i$). 
Initially, the tree is empty. Let $\textsc{Add}(i)$ be a function that adds a string $s^i$ to the tree. Let $\textsc{GetMin}$ be a function that returns the index of the minimal string $s$ from the tree according to the comparing procedure. After returning the index, the function removes it from the tree.
The final algorithm is presented as Algorithm \ref{alg:sortstring} and its complexity in Theorem \ref{th:upper-sort}.
\begin{algorithm}[ht]
    \caption{Quantum string sorting algorithm} \label{alg:sortstring}
    \begin{algorithmic}
        \For{$i\in\{0,\dots,n-1\}$}
            \State $\textsc{Add}(i)$
        \EndFor
        \For{$i\in\{0,\dots,n-1\}$}
            \State $j_i\gets \textsc{GetMin}$
        \EndFor
    \end{algorithmic}
\end{algorithm}
The second For-loop can be replaced by inorder traversal (dfs) of the tree for constructing the list. This approach has smaller hidden constant in big-O. The full idea is presented in Appendix \ref{apx:sort-with-dfs} for completeness. 
\begin{theorem}\label{th:upper-sort}
  The quantum running time for sorting $n$ string of size $l$ is $O(n\sqrt{l}\log n)$ and $\Omega(n\sqrt{l}/\log n)$.
\end{theorem}
The upper bound is complexity of the presented algorithm and algorithm from \cite{kiv2022}. The proof of the lower bound is presented in the next section. A reader can see that the difference between upper and lower bounds is just $(\log n)^{1.5}$. It shows that the presented algorithm is almost optimal.

\subsection{Lower Bound}\label{sub:lower-bound}
For simplicity we assume that strings are binary, i.e. $s^i\in\{0,1\}^l$, for $i\in\{0,\dots,n-1\}$.
Let us formally define the sorting function.

  For  positive integers $n,l$,
 let \(\RADIX_{n,l}: \{0,1\}^{nl} \to \mathbb{S}_{n}\) be a function that gets \(n\) binary
  strings of length \(l\) as input and returns a permutation of \(n\)
  integers that is a result of sorting input strings. Here $\mathbb{S}_{n}$ is a set of all permutations of integers from $0$ to $n-1$.
  For \(s^{0}, \ldots, s^{n-1} \in \{0,1\}^{l}\), we have  $
    \RADIX_{n,l}(s^{0}, \ldots, s^{n-1}) = (j_{0}, \ldots, j_{n-1}),
  $
  where \((j_{0}, \ldots, j_{n-1}) \in \mathbb{S}_{n}\) is a permutation such that 
  \(\sigma_{j_{i}} < \sigma_{{j_{i+1}}}\) or (\(\sigma_{j_{i}} = \sigma_{j_{i+1}}\) and \(j_{i} < j_{i+1}\)),  for \(i \in \{0, \ldots, n-2\}\). 
  
Note that in the case of \(l = 1\), the function \(\RADIX_{n,1}\) can be used to
compute the majority function. We use \(\RADIX_{n,1}\) to sort strings and the
\(n/2\)-th string is a value of the majority function. Therefore, we expect that
complexity of \(\RADIX_{n,1}\) should be \(\Omega(n)\) \cite{BHCMW98}.
In the case of \(n = 2\), the function \(\RADIX_{2,l}\) is similar to the OR function,
so we expect that it requires \(\Omega(\sqrt{l})\) queries \cite{BHCMW98}. The formal proof for relations between these functions is presented in Appendix \ref{apx:radix-maj-firstone}.


%


We denote \(\|M\|\) the spectral norm of a matrix \(M\), and \(A \circ B\)
denotes the Hadamard product of matrices \(A\) and \(B\):
\((A \circ B)(\sigma, \tau) = A(\sigma, \tau) B(\sigma, \tau)\).


We prove a lower bound for \(\RADIX\) using Adversary method \cite{hls2007}.

\begin{theorem}[Adversary bound, \cite{hls2007}]\label{th:adv-lowerbound}
  Let \(f : \{0,1\}^{n} \to \Sigma_{O}\) be an arbitrary function, where
  \(\Sigma_{O}\) is a set of outputs.
  Let \(A\) be an arbitrary matrix with rows and columns indexed by input
  strings, such that \(A(\sigma, \tau) = 0\) if \(f(\sigma) = f(\tau)\).
  Let \(D_{i}\) be a zero-one matrix with rows and columns indexed by input
  strings, such that
  \(D_{i}(\sigma, \tau) = 1\) iff these strings differ in \(i\)-th symbol: \(\sigma_{i} \neq \tau_{i}\).
  Let us denote
$
    \ADV(f) = \max_{A \ge 0} \frac{\| A \|}{\max_{i} \|A \circ D_{i}\|},
$
  where the maximum is taken over all matrices \(A\) with non-negative entries.
  Then the two-sided \(\epsilon\)-bounded error quantum query complexity
  \(Q(f)\) is lower bounded by
 $   Q(f) \ge \frac{1-2\sqrt{\epsilon(1-\epsilon)}}{2} \ADV(f).$
\end{theorem}

Bounds obtained by using
Adversary method can be composed.
It is a  nice property that we use.
\begin{corollary}[\cite{hls2007}]\label{thm:adv-compose}
  If \(h = f \circ (g_{1}, \ldots, g_{k})\), subfunctions \(g_{k}\) act on disjoint subsets of input, and every \(g_{k} = g\), then
  $  \ADV(h) = \ADV(f) \ADV(g)$.
\end{corollary}

To obtain an upper bound on \(\|A \circ D_{i}\|\), the following lemma is used.
\begin{lemma}[\cite{ss2006adv}]\label{lem:circ-prod}
  Assume \(A\), \(B\) and \(C\) are real matrices such that \(A = B \circ C\). Then,
    $\|A\| \le \max_{i,j : A(i,j)\neq 0} r_{i}(B)c_{j}(C)$,
  where \(r_{i}(B)\) is the \(\ell_{2}\)-norm of the \(i\)-th row of \(B\), and \(c_{j}(C)\) is the \(\ell_{2}\)-norm of the \(j\)-th column of \(C\).
\end{lemma}



Firstly, let us note that sorting a list cannot be easier than computing the sign of permutation that sorts that list
. We prove it in Lemma~\ref{lem:median}.
%
  For positive integers \(n,l\),
  let \(\SIGNUM_{n,l}: \{0,1\}^{nl} \to \{0,1\}\) be a function that for   \(n\) binary strings of length
  \(l\) returns whether the permutation that sorts the input list is odd or even.
  Let \(s^{0}, \ldots, s^{n-1} \in \{0,1\}^{l}\) be input strings. Then
 $
    \SIGNUM_{n,l}(s^{0}, \ldots, s^{n-1}) = \sgn(\RADIX_{n,l}(s^{0}, \ldots, s^{n-1})
 $. Here $\sgn(j_0,\dots,j_{n-1})$ is $0$ if the permutation $(j_0,\dots,j_{n-1})$ is even, and $1$ otherwise. A permutation is even iff the number of inversions in the permutation is even. See \cite{Princeton} for more details.

\begin{lemma}\label{lem:median}
  Computing \(\RADIX_{n,l}\) is not easier than computing \(\SIGNUM_{n,l}\):
  
  $
    \QC(\RADIX_{n,l}) \ge \QC(\SIGNUM_{n,l}).
  $ (See Appendix \ref{apx:median}.)
\end{lemma}

We restrict input for \(\SIGNUM\) as follows (restricting input cannot make
computing \(\SIGNUM\) easier).
Let us interpret an \(n \times l\) zero-one input matrix as a combination of \(l / \log_2 n\) bands of zeroes and ones of size \(n \times \log_2 n\). We can interpret each band as a list of \(n\) numbers from \(0\) to \(n-1\), because \(\log_2 n\) bits are enough to encode these numbers.
Let us call a band {\it proper} if it either contains a permutation (e.g. it contains each number from \(0\) to \(n-1\) exactly once) or it contains only zeroes.
We consider only inputs for \(\SIGNUM\) that are composed from proper bands.
%
%
  Let \(\SG_{n}: \{0,1\}^{n\log n} \to \{0,1,\bot\}\) be a function that gets
  a proper band of \(n\) words of \(\log n\) symbols and returns either a sign
  of permutation, or \(\bot\) if all words are equal. Here by a sign of permutation, we mean the result of $sgn$ function.
%
  Let \(\FST_{m} : \{0,1,\bot\}^{m} \to \{0,1\}\) be a function that gets \(m\)
  symbols and returns the first symbol that is not equal to \(\bot\). If all
  arguments are equal to \(\bot\), then \(\FST_{m}\) returns \(0\).
%
It is easy to see that the following holds
\begin{lemma}\label{lm:signum}
$\SIGNUM_{n,l} = \FST_{l/\log n} \circ \big( \SG_{n}, \ldots, \SG_{n} \big),$
where each \(\SG_{n}\) acts on its own band. (See Appendix \ref{apx:signum}).
\end{lemma}

We can obtain Adversary bound for the introduced functions.
\begin{lemma}\label{lem:adv-fst} We have
 $
    \ADV(\FST_{m}) = \Omega(\sqrt{m}).
  $ (See Appendix \ref{apx:adv-fst}).
\end{lemma}


\begin{lemma}\label{lem:sg-fst}  We have
 $
    \ADV(\SG_{n}) = \Omega(n)
  $. (See Appendix \ref{apx:sg-fst}).
\end{lemma}


Finally, we apply Theorem \ref{thm:adv-compose} to Lemmas \ref{lem:adv-fst}
and \ref{lem:sg-fst}.
\begin{lemma}\label{thm:adv-signum}
 The following statement is right  $
    \ADV(\SIGNUM_{n,l}) = \Omega(n\sqrt{l/\log n}).
 $
\end{lemma}
Using this lemma and Theorem \ref{th:adv-lowerbound} we obtain the lower bound for the string sorting problem.
\begin{theorem} The following statement is right 
  $
    \QC(\RADIX_{n,l}) = \Omega(n\sqrt{l/\log n})
 $.
\end{theorem}


\section{Auto-Complete Problem}\label{sec:appl}

\paragraph{Problem.} Assume that we use some constant size alphabet, for example, binary, ASCII or Unicode. We work with a sequence of strings $\cS=(s^{i_1},\dots,s^{i_{|\cS|}})$, where $|\cS|$ is the length of the sequence; and $i_j$ are increasing indexes of strings. In fact, the index $i_j$ is the index of the query for adding this string to $\cS$.
Initially, the sequence $\cS$ is empty.
Then, we have $n$ queries of two types.
The first type is adding a string $s$ to the sequence $\cS$. Let $\#(u)=|\{j: u=s^j, j\in\{i_1,\dots,i_{|\cS|}\}\}|$ be a number of occurrence (or ``frequency'') of a string $u$ among $\cS$.
The second type is querying the most frequent complement from $\cS$ of a string $t$.
 Let us define it formally. 
 If $t$ is a prefix of $s^j$, then we say that $s^j$ is a complement of $t$.
 Let $D(t)=\{s^j:j\in\{i_1,\dots,i_{|\cS|}\}, s^j$ is a complement of $t\}$ be the set of complements for $t$, and let $md(t)=max\{\#(s^r): r\in\{i_1,\dots,i_{|\cS|}\}, s^r\in D(t)\}$ be the maximal ``frequency'' of strings from $D(t)$. The problem is to find the index $mi(t)=min\{r:r\in\{i_1,\dots,i_{|\cS|}\},  s^r\in D(t), \#(s^r)=md(t)\}$.

Here we present the main idea of a solution and a full description of the algorithm is presented in Appendix \ref{apx:complement}.

We use a Self-Balanced Search tree for our solution. A node $v$ of the tree stores 4-tuple $(i,c,j,jc)$, where $i$ is an index of a string $s^i$ that is ``stored'' in the node, and $c=\#(s^i)$. The tree is a Search tree by this strings $s^i$ similar to storing strings in Section \ref{sec:sort}. For comparing strings, we use a quantum procedure from Lemma \ref{lm:strcmp}. Therefore, our tree is noisy.
The index $j$ is the index of the most ``frequent'' string in the sub-tree which root is $v$, and $jc=\#(s^j)$. Formally, for any vertex $v'$ from this sub-tree if $(i',c',j',jc')$ is associated with $v'$, then $c'<jc$ or ($c'=jc$ and $i'\geq j$).

Initially, the tree is empty. Let us discuss processing the first type of queries. We want to add a string $s$ to $\cS$. We search a node $v$ with associated $(i,c,j,jc)$ such that $s^i=s$. If we can find it, then we increase $c\gets c+1$. It means $j$ parameter of the node $v$ or its ancestors can be updated. There are at most $O(\log n)$ ancestors because height of the tree is $O(\log n)$. So, for each  ancestor of $v$ that associated with $(i',c',j',jc')$, if $jc'<c$ or ($jc'=c$ and $j'>i$), then we update $j'\gets i$ and $jc\gets c$.

If we cannot find the string $s$ in $\cS$, then we added a new node to the tree with associated 4-tuple $(r,1,r,1)$, where $r$ is the index of the query. Note that if we re-balance nodes in the red-black tree, then we easily can recompute $j$ and $jc$ elements of nodes. 

Let us discuss processing the second type of queries. All strings $s^i$ that can be complement for $t$ belongs to the set $C(t)=\{s^r: r\in\{i_1,\dots,i_{|\cS|}\}, s^r\geq t$ and $s^r<t'$\}. Here we can obtain $t'$ from the string $t$ by replacing the last symbol by the next symbol from the alphabet. Formally, if $t=(t_1,\dots,t_{|t|-1},t_{|t|})$, then $t'=(t_1,\dots,t_{|t|-1},t'_{|t|})$, where the symbol $t'_{|t|}$ succeeds $t_{|t|}$ in the alphabet. We can say, that $C(t)=D(t)$. The query processing consist of three steps.

\textbf{Step 1.} We search for a node $v_{Rt}$ such that $t$ should be in the left sub-tree of $v$ and $t'$ should be in the right sub-tree of $v$. Formally, $\beta(v_{Rt})\leq t\leq \alpha(v_{Rt})$, and $\alpha(v_{Rt})< t'\leq \gamma(v_{Rt})$. In the next two steps, we compute the index of the required string $j_{ans}$ and $jc_{ans}=\#(s^{j_{ans}})$.

\textbf{Step 2.} Let us look to the left sub-tree with $t$. Let us find a node $v_L$ that contains an index $i_L$ of  the minimal string $s^{i_L}\geq t$. Then, we go up from this node. Let us consider a vertex $v$. If it is the right child of $\textsc{Parent}(v)$, then it is bigger than the string from $\textsc{Parent}(v)$ and the left child's sub-tree, so we do nothing. If it is the left child of $\textsc{Parent}(v)$, then it is less than the string from $\textsc{Parent}(v)$ and all strings from the right child's sub-tree, so we update $j_{ans}$ and $jc_{ans}$ by values from the parent node and the right child. Formally, if $(i_p,c_p,j_p,jc_p)$ for the $\textsc{Parent}(v)$ node and $(i_r,c_r,j_r,jc_r)$ for the right child node, then we do the following actions. If $c_p>jc_{ans}$ or ($c_p=jc_{ans}$ and $i_p<j_{ans}$), then $j_{ans}\gets i_p$ and $jc_{ans}\gets c_p$.    If $jc_r>jc_{ans}$ or ($jc_r=jc_{ans}$ and $j_r<j_{ans}$), then $j_{ans}\gets j_r$ and $jc_{ans}\gets jc_r$.    

\textbf{Step 3.} Let us look to the right sub-tree with $t'$. Let us find a node $v_R$ that contains an index $i_R$ of  the maximal string $s^{i_R}< t'$. Then, we go up from this node and do the symmetric actions as in Step 2.

Each of these three steps requires $O(\log n)$ running time because each of them observes nodes of a single branch. Finally, we obtain quantum complexity of the problem.
\begin{theorem}\label{th:comlete-quantum}
The quantum algorithm with a noisy Self-Balanced Search tree for Auto-Complete Problem has $O(\sqrt{nL}\log n)$ running time and error probability $1/3$, where $L$ is the sum of lengths of all queries. Additionally, the lower bound for quantum running time is $\Omega(\sqrt{L})$. (See Appendix \ref{apx:complement}).
\end{theorem}
Let us consider the classical (deterministic or randomize) case. If we use the same Self-balanced search tree, then the running time is $O(L\log n)$. At the same time, if we use the Trie (prefix tree) data structure 
\cite{knuth73}, then the complexity is $O(L)$.
We can show that it also the lower bound for classical case.
\begin{lemma}\label{lm:comlete-classical}
The classical running time for Auto-Complete Problem  is $\Theta(L)$, where $L$ is the sum of lengths of all queries. (See Appendix \ref{apx:complement}).
\end{lemma}
If $O(n)$ strings  have length at least $\Omega((\log_2 n)^2)$, then we obtain a quantum speed-up.

\section{Conclusion}\label{sec:concl}
We suggest the Walking tree technique for noisy tree data structures. We apply the technique to the Red-black tree and the Segment tree. We show that the complexity of the main operations is asymptotically equal to the complexity of standard (not noisy) versions of the data structures.
We use noisy Red-black tree for two problems: The Strings Sorting Problem and The Auto-Complete Problem. The considered algorithms are quantum because they use the quantum string comparing algorithm as a subroutine. This subroutine demonstrates quadratic speed-up, but it has a non-zero error probability. For the String Sorting Problem, we show lower and upper bounds that are the same up to a log factor. For Auto-Complete Problem, we obtain quantum speed-up for the problems if $(\log n)^2=o(l)$ where $n$ is the number of queries and $l$ is the size of an input string of a query. 

Future work can be applying the Walking tree technique to other tree data structures and obtain the noisy version of them with good complexity of main operations. Also, it is interesting to find more applications for noisy data structures. We assume that quantum algorithms should be one of the fruitful fields for such applications. It is interesting to meet quantum lower and upper bounds for considered problems.
%
%
%
 \bibliographystyle{splncs04}
 \bibliography{main}
\newpage
\appendix
\section{Proof of Theorem \ref{th:wt-compl}}\label{apx:wt-compl}
\textbf{Theorem \ref{th:wt-compl}.} {\em
Suppose, $s=O(h+\log(1/\varepsilon))$. Then, Algorithm \ref{alg:walking} does the same action as Algorithm \ref{alg:maincomand} with error probability $\varepsilon$.}
\begin{proof}
Let us consider the walking tree. We emulate the counter by replacing it with a nodes chain of length $s+1$.
Formally, for a node $u\in\V (W)$ we add $s+1$ nodes $d^u_1,\dots,d^u_{s+1}$ such that $\textsc{Parent}(d^u_1)=u$, $\textsc{Parent}(d^u_i)=d^u_{i-1}$ for $i\in\{2,\dots,s+1\}$. The only child of $d^u_i$ is $d^u_{i+1}$ for $i\in\{1,\dots,s\}$ and $d^u_{s+1}$ does not have children.

In that case the increasing of $c(u)$ can be emulated by moving from $d^u_{c(u)}$ to $d^u_{c(u)+1}$. The decreasing can be emulated by moving from $d^u_{c(u)}$ to $d^u_{c(u)-1}$. We can assume that $d^u_0$ is the node $u$ itself.

Let $u_{target}$ be the target node, i.e. $\textsc{IsItTheTarget}(u_{target})=True$. Let us consider the distance $L$ between the target node $d^{u_{target}}_{s+1}$ and the current node in the modified tree. The distance $L$ is a random variable. Each step of the walk increase or decrease the distance $L$ by $1$. So, we can present $L=dist(root,d^{u_{target}}_{s+1})-Y$, where  $root$ is the root node of $W$, $Y=(Y_1+\dots+Y_s)$, and $Y_i\in\{-1,+1\}$ are independent random variables that represent $i$-th step and show increasing or decreasing the distance. Let $Y_i=+1$ if we move in the correct direction, and $Y_i=-1$ if we move in the wrong direction. Note that the probability of moving to the correct direction ($Y_i=+1$) is at least $2/3$ and the probability of moving to the wrong direction ($Y_i=-1$) is at most $1/3$.
From now on without loss of generality, we assume that $\Prb{Y_i=+1}=2/3$ and $\Prb{Y_i=-1}=1/3$.

If $L\leq s$, then we are in the $d^{u_{target}}_{i}$ node in the modified tree and in the $u_{target}$ node in the original walking tree $W$.
Note that $dist(root,d^{u_{target}}_{s+1})\leq h+s$, where $h=h(W)$ by the definition of the height of a tree.
Therefore, $L\leq s$ means $Y=Y_1+\dots+Y_{s}\geq h$. So, the probability of success of the operation is the probability of the $Y\geq h$ event, i.e. $\mathrm{\mathbf{Pr}}_{success}=\Prb{Y\geq h}$.

Let $X_i=(Y_i+1)/2$ for $i\in\{1,\dots,s\}$. We treat $X_1,\dots,X_s$ as independent binary random variables. Let $X=\sum_{i=1}^s X_i$. For such $X$ and for any $0<\delta\leq1$, the following form of Chernoff bound \cite{mr2006} holds
\begin{equation}\label{eq:chernoff}
\Prb{X<(1-\delta)\Exp{X}}\leq \exp(-\Exp{X}\delta^2/2).    
\end{equation}

Since $\Prb{X_i=1}=\Prb{Y_i=+1}=2/3$, we have $\Exp{X}=2s/3$ and the inequality (\ref{eq:chernoff}) becomes
$$\Prb{X<2s\cdot(1-\delta)/3}\leq \exp(-\delta^2s/3).$$

Substituting $Y$ for $X$ we get
$$\Prb{Y<s\cdot(4(1-\delta)/3-1)}\leq \exp(-\delta^2s/3).$$

From now on without loss of generality, we assume that $s=\ceil{c\cdot (h+\log(1/\varepsilon))}$ for some $c>0$. Let $\delta=\frac{1}{6}$ and $c\geq108$.

In the following steps, we relax the inequality by obtaining less tight bounds for the target probability.

Firstly, we obtain a new lower bound
\begin{equation*}
\begin{split}
s\cdot(4(1-\delta)/3-1)&=\ceil{c\cdot(h+\log(1/\varepsilon))}/9\geq c\cdot(h+\log(1/\varepsilon))/9\geq\\
&\geq c\cdot h/9\geq 108\cdot h/9>h,
\end{split}
\end{equation*}
and hence
$$\Prb{Y<h}\leq\Prb{Y<s\cdot(4(1-\delta)/3-1)}$$

Secondly, we obtain a new upper bound
\begin{equation*}
\begin{split}
\exp(-\delta^2s/3)&=\exp(-\ceil{c\cdot(h+\log(1/\varepsilon))}/108)\leq\exp(-c/108\cdot(h+\log(1/\varepsilon))<\\&<\exp(-c/108\cdot\log(\varepsilon^{-1}))\leq \exp(-\log(\varepsilon^{-1}))=\varepsilon.
\end{split}
\end{equation*}

Combining the two obtained bounds we have
$$\Prb{Y<h}\leq\Prb{Y<(\frac{4}{3}(1-\delta)-1)s}\leq \exp(-\delta^2s/3)<\varepsilon,$$
and hence
$$\Prb{Y<h}<\varepsilon.$$

Considering the probability of the opposite event we finally get
$$\mathrm{\mathbf{Pr}}_{success}=\Prb{Y\geq h}>1-\varepsilon.$$
\Endproof \end{proof}

\section{Quantum Sorting Algorithm. Second Approach}\label{apx:sort-with-dfs}
Assume that we want to construct have a result list. We use $list$ variable as a result.

We use the recursive procedure $\textsc{GetListByTree}(v) $ for in-order traversal (dfs) of the searching tree. Here $v$ is the processing node. Assume that we have $\textsc{GetTheLeftChild}(v) $ for obtaining the left child of $v$, $\textsc{GetTheRightChild}(v) $ for obtaining the right child of $v$; $\textsc{GetIndexOfString}(v) $ for obtaining the index of the string $\alpha(v)$ that is stored in $v$. The procedure is presented in Algorithm \ref{alg:sort-dfs}.

\begin{algorithm}[ht]
    \caption{The recursive procedure $\textsc{GetListByTree}(v) $ for in-order traversal (dfs) of the searching tree} \label{alg:sort-dfs}
    \begin{algorithmic}
        \If{$v \neq NULL$}
        \State $\textsc{GetListByTree}(\textsc{GetTheLeftChild}(v)) $
        \State $list\gets list \circ \textsc{GetIndexOfString}(v) $
        \State $\textsc{GetListByTree}(\textsc{GetTheRightChild}(v)) $
        \EndIf
    \end{algorithmic}
\end{algorithm}

The total sorting algorithm is Algorithm \ref{alg:sort2}.
\begin{algorithm}[ht]
    \caption{Quantum string sorting algorithm} \label{alg:sort2}
    \begin{algorithmic}
        \For{$i\in\{0,\dots,n-1\}$}
            \State $\textsc{Add}(i)$
        \EndFor
        \State $list\gets []$\Comment{Initially, the list is empty}
         \State $\textsc{GetListByTree}(\textsc{GetTheRoot}())$
        \For{$i\in\{0,\dots,n-1\}$}
            \State $j_i\gets list[i]$
        \EndFor
    \end{algorithmic}
\end{algorithm}
\section{Complexities of \(\RADIX_{n,1}\) and \(\RADIX_{2,l}\)}\label{apx:radix-maj-firstone}

This section contains a formal statement and proof for the note in the beginning of Section~\ref{sub:lower-bound} that reads as follows.

Note that in the case of \(l = 1\), the function \(\RADIX_{n,1}\) can be used to
compute the majority function. We use \(\RADIX_{n,1}\) to sort strings and the
\(n/2\)-th string is a value of the majority function. Therefore, we expect that
complexity of \(\RADIX_{n,1}\) should be \(\Omega(n)\).
In the case of \(n = 2\), the function \(\RADIX_{2,l}\) is similar to the OR function,
so we expect that it requires \(\Omega(\sqrt{l})\) queries.

Formally, we prove the following
\begin{lemma}
For positive integers \(n,l\), let \(\mathrm{MAJ}_n: \{0,1\}^n \to \{0,1\}\) be a majority function, and \(\mathrm{SEARCH1}_l: \{0,1\}^l \to \{1,\ldots,l\}\) be a function that returns the minimal index of one in the input.
\begin{align}
\QC(\RADIX_{n,1}) &\ge \QC(\mathrm{MAJ}_n) \label{eq:radix-n1-maj}\\
\QC(\RADIX_{2,l}) &\ge \QC(\mathrm{OR}_l) \label{eq:radix-2l-or}
\end{align}
\end{lemma}

\begin{proof}
Consider the input \(x \in \{0,1\}^n\). Suppose that \(\RADIX_{n,1}(x) = \sigma = (\sigma_1, \ldots, \sigma_n)\).

Then the proof of~\eqref{eq:radix-n1-maj} follows from the fact that
\[
\mathrm{MAJ}_n(x) = x_{\sigma_{\lceil n/2 \rceil}}.
\]

Take an input \(y \in \{0,1\}^l\). Let \(y'\in \{0,1\}^{2l}\) be a pair of words \(\vec{0} = 0\ldots 0 \in \{0,1\}^l\) and \(y\). Now we see that
\[
\mathrm{OR}_l(y) = \begin{cases} 
1, & \quad \text{if } \RADIX_{2,l}(y') = (1 0), \\ 
0, & \quad \text{if } \RADIX_{2,l}(y') = (0 1).
\end{cases}
\]
and this completes the proof.
\end{proof}
\section{Proof of Lemma \ref{lem:median}}\label{apx:median}
\textbf{Lemma \ref{lem:median}.} \textit{Computing \(\RADIX_{n,l}\) is not easier than computing \(\SIGNUM_{n,l}\):
  $
    \QC(\RADIX_{n,l}) \ge \QC(\SIGNUM_{n,l}).
  $}

\begin{proof}
  We can compute \(\SIGNUM\) by computing the sign of the permutation that is a result of computing \(\RADIX\). It
  requires no call to the oracle, we can simulate the classical algorithm for computing the sign of the permutation.
  \Endproof
\end{proof}
\section{Proof of Lemma \ref{lm:signum}}\label{apx:signum}
\textbf{Lemma \ref{lm:signum}.} \textit{
$\SIGNUM_{n,l} = \FST_{l/\log n} \circ \big( \SG_{n}, \ldots, \SG_{n} \big),$
where each \(\SG_{n}\) acts on its own band.
}
\begin{proof}
  Consider an input \(w = (w_{1}, w_{2}, \ldots, w_{l/\log n})\) for \(\SIGNUM_{n,l}\), where each \(w_{i}\) is a proper
  band, i.e. either \(w_{i} = \vec{0} = (0, \ldots, 0)\) or
  \(w_{i} = \sigma = (\sigma_{i_1}, \sigma_{i_2}, \ldots, \sigma_{i_n}) \in \mathbb{S}_{n}\).

  The statement of the lemma is that
  \[\SIGNUM_{n,l}(w) = \FST_{l/\log n}(\SG_{n}(w_{1}), \SG_{n}(w_{2}), \ldots, \SG_{n}(w_{l/\log n}))\]
  for every choice of proper bands \(w_{i}\).

  If input is all-zero, i.e. \(w_{1} = w_{2} = \ldots = w_{n} = \vec{0}\), then \(\SIGNUM_{n,l} = 0\) and
  \(\SG_{n}(w_{i}) = \bot\) for \(i \in \{1,2,\ldots,l/\log n\}\). Therefore, the statement is correct, because
  \(\FST_{l/\log n}(\bot, \ldots, \bot) = 0\) by definition.

  Suppose that input is not all-zero.

  Let \(j\) be a minimal \(i\), such that \(w_{i} \neq \vec{0}\). Let us denote \(w_{j} = \sigma \in \mathbb{S}_{n}\).
  Therefore, \(w = (\vec{0}, \ldots, \vec{0}, \sigma, \ldots)\). The result of sorting \(w\) is completely defined by
  \(\sigma\), because \(\sigma\) contains each number from \(0\) to \(n-1\) exactly once:
  \(\RADIX_{n,l}(w) = \sigma^{-1}\).

  So, on the one hand, we have \(\SIGNUM_{n,l}(w) = \sgn(\RADIX_{n,l}(w)) = \sgn(\sigma^{-1}) = \sgn(\sigma)\).

  On the other hand, \(\SG_{n}(w_{i}) = \bot\) for \(i < j\) and \(\SG_{n}(w_{j}) = \sgn(w_{j}) = \sgn(\sigma)\).
  Therefore,
  \[\FST_{l/\log n}(\SG_{n}(w_{1}), \ldots, \SG_{n}(w_{l/\log n})) = \FST_{l/\log n}(\bot, \ldots, \bot, \sgn(\sigma), \ldots) = \sgn(\sigma),\]
  and the lemma is proved.
  \Endproof
\end{proof}

\section{Proof of Lemma \ref{lem:adv-fst}}\label{apx:adv-fst}
\textbf{Theorem \ref{lem:adv-fst}.} \textit{
The following statement is right:
 $
    \ADV(\FST_{m}) = \Omega(\sqrt{m}).
  $
}
\begin{proof}
  Let us denote \(\bot^{m} = (\bot, \ldots, \bot)\) a string that contains only \(\bot\) symbols and let us denote \(\bot^{m}_{i \mapsto x} = (\bot, \ldots, \bot, x, \bot, \ldots, \bot)\) a string that contains \(\bot\) symbols at all positions except position \(i\) that is \(x\).

  Let us define a matrix \(A\) with rows and columns indexed by strings from \(\mathcal{I} = \{0,1,\bot\}^{m}\) as follows.
  \[
    A(\sigma, \tau) =
    \begin{cases}
      1, \quad \text{ if } \sigma = \bot^{m} \text{ and } \tau = \bot^{m}_{i \mapsto 1} \text{ for some } i \in \{1,\ldots,m\},\\
      0, \quad \text{ otherwise}
    \end{cases}
  \]

  \[
    \|A\| = \max_{x \neq 0}\frac{\|Ax\|}{\|x\|} \ge \frac{\|Ax_{0}\|}{\|x_{0}\|} = \frac{m-1}{\sqrt{m-1}} = \sqrt{m-1},
  \]
  where \(x_{0}\) is a zero-one vector indexed by \(\mathcal{I}\), such that \(x_{0}(\tau) = 1\) iff \(\tau = \bot^{m}_{i \mapsto 1}\) for some \(i \in \{1,\ldots,m\}\).

  Let us fix some \(i \in \{1,\ldots,m\}\). Then
  \[
    (A \circ D_{i})(\sigma, \tau) =
    \begin{cases}
      1, \quad \text{ if } \sigma = \bot^{k} \text{ and } \tau = \bot^{k}_{i \mapsto 1},\\
      0, \quad \text{ otherwise}
    \end{cases}
  \]

  Therefore, \(\|A \circ D_{i}\| = 1\) and \(\ADV(\FST_{m}) = \sqrt{m-1}\).
  \Endproof
\end{proof}
\section{Proof of Lemma \ref{lem:sg-fst}}\label{apx:sg-fst}
\textbf{Lemma \ref{lem:sg-fst}.} \textit{
We have
 $
    \ADV(\SG_{n}) = \Omega(n)
  $.
}
\begin{proof}
  Let us denote possible \(n \times \log n\) input matrices for \(\SG_{n}\) by integer vectors of size \(n\) and let us denote \(\mathcal{I} = \{\vec{0}\} \cup \mathbb{S}_{n} \), where \(\vec{0}=(0,\ldots,0)\).

  We use latin letters \(x,y,\ldots\) to denote any element of \(\mathcal{I}\) and greek letters \(\sigma, \tau, \ldots\) to denote elements of \(\mathbb{S}_{n}\).

  Let us define a matrix \(A\) with rows and columns indexed by \(\mathcal{I}\) as follows.
  \begin{align*}
    A(\vec{0}, y) &= A(x, \vec{0}) = 0, \\
    A(\sigma, \tau) &= \begin{cases}
      1, \quad \text{ if } \sgn(\sigma) = 0, \sgn(\tau) = 1, \sigma = \tau \circ (s t) \text{ for some } s,t \in [n], \\
      0, \quad \text{ otherwise}.
    \end{cases}
  \end{align*}

  Each row of \(A\) has \(n(n-1)\) non-zero entries, so we obtain \(\|A\| \ge n(n-1) \).

  Let us fix some \(i \in [n\log n]\). Denote \(B_{i} = A \circ D_{i}\).

  Fixing \(i\)-th bit in the input matrix means that we fix \(k\)-th bit in \(j\)-th entry in the input permutation \(\rho \in \mathbb{S}_{n}\). Consider an entry \(B_{i}(\sigma, \tau)\). It is non-zero iff \(\sigma = \tau \circ (s t)\). At the same time, \(\sigma_{j}\) and \(\tau_{j}\) have to differ in \(k\)-th bit, so \((s t)\) has to replace \(j\)-th entry with one of the \(n/2\) entries that differ in \(k\)-th bit.

  Each row and each column of \(B_{i}\) has \(n/2\) non-zero entries. Therefore, by Lemma~\ref{lem:circ-prod} we have
  \(\|(A \circ D_{i})\| \le n/2\), and \(\ADV(\SG_{n}) = \Omega(n)\).
  \Endproof
\end{proof}
\section{Algorithms and Lower Bounds for Auto-Complete Problem}\label{apx:complement}

Here we present the full description of solution for the Auto-Complete Problem.
\paragraph{Problem.} Assume that we use some constant size alphabet, for example, binary, ASCII or Unicode. We work with a sequence of strings $\cS=(s^{i_1},\dots,s^{i_{|\cS|}})$, where $|\cS|$ is the length of the sequence; and $i_j$ are increasing indexes of strings. In fact, the index $i_j$ is the index of the query for adding this string to $\cS$.
Initially, the sequence $\cS$ is empty.
Then, we have $n$ queries of two types.
The first type is adding a string $s$ to the sequence $\cS$. Let $\#(u)=|\{j: u=s^j, j\in\{i_1,\dots,i_{|\cS|}\}\}|$ be a number of occurrence (or ``frequency'') of a string $u$ among $\cS$.
The second type is querying the most frequent complement from $\cS$ of a string $t$.
 Let us define it formally. 
 If $t$ is a prefix of $s^j$, then we say that $s^j$ is a complement of $t$.
 Let $D(t)=\{s^j:j\in\{i_1,\dots,i_{|\cS|}\}, s^j$ is a complement of $t\}$ be the set of complements for $t$, and let $md(t)=max\{\#(s^r): r\in\{i_1,\dots,i_{|\cS|}\}, s^r\in D(t)\}$ be the maximal ``frequency'' of strings from $D(t)$. The problem is to find the index $mi(t)=min\{r:r\in\{i_1,\dots,i_{|\cS|}\},  s^r\in D(t), \#(s^r)=md(t)\}$.

We use a Self-Balanced Search tree for our solution. A node $v$ of the tree stores 4-tuple $(i,c,j,jc)$, where $i$ is an index of a string $s^i$ that is ``stored'' in the node, and $c=\#(s^i)$. The tree is a Search tree by this strings $s^i$ similar to storing strings in Section \ref{sec:sort}. For comparing strings, we use a quantum procedure from Lemma \ref{lm:strcmp}. Therefore, our tree is noisy.
The index $j$ is the index of the most ``frequent'' string in the sub-tree which root is $v$, and $jc=\#(s^j)$. Formally, for any vertex $v'$ from this sub-tree if $(i',c',j',jc')$ is associated with $v'$, then $c'<jc$ or ($c'=jc$ and $i'\geq j$).

Initially, the tree is empty. Let us discuss processing the first type of queries. We want to add a string $s$ to $\cS$. We search a node $v$ with associated $(i,c,j,jc)$ such that $s^i=s$. If we can find it, then we increase $c\gets c+1$. It means $j$ parameter of the node $v$ or its ancestors can be updated. There are at most $O(\log n)$ ancestors because height of the tree is $O(\log n)$. So, for each  ancestor of $v$ that associated with $(i',c',j',jc')$, if $jc'<c$ or ($jc'=c$ and $j'>i$), then we update $j'\gets i$ and $jc\gets c$.

If we cannot find the string $s$ in $\cS$, then we added a new node to the tree with associated 4-tuple $(r,1,r,1)$, where $r$ is the index of the query. Note that if we re-balance nodes in the red-black tree, then we easily can recompute $j$ and $jc$ elements of nodes. 

Assume that we have a $\textsc{Search}(s)$ procedure that returns a node $v$ with associated $(i,c,j,jc)$ where $s=s^i$. If there is no such a node, then the procedure returns $NULL$. A procedure $\textsc{AddAString}(r)$ adding a node $(r,1,r,1)$ to the tree. A procedure $\textsc{GetTheRoot}$ returns the root of the search tree. The processing of the first type of queries is presented in Algorithm \ref{alg:compl1} 
\begin{algorithm}[ht]
    \caption{Processing a query of the first type with an argument $s$ and a query number $r$} \label{alg:compl1}
    \begin{algorithmic}
        \State $v \gets \textsc{Search}(s)$
        \If{$v\neq NULL$}
         \State $(i,c,j,jc)$ is associated with $v$
         \State $c\gets c+1$
          \If{$jc<c$ or $(jc=c$ and $j\geq i)$}
         \State $jc\gets c$
         \State $j\gets i$
         \EndIf
         \While{$v\neq \textsc{GetTheRoot}$}
         \State $v\gets \textsc{Parent}(v)$
         \State $(i',c',j',jc')$ is associated with $v$ 
         \If{$jc'<c$ or $(jc'=c$ and $j'\geq i)$}
         \State $jc'\gets c$
         \State $j'\gets i$
         \EndIf
         \EndWhile
        \Else
         \State $\textsc{AddAString}(r)$
        \EndIf
    \end{algorithmic}
\end{algorithm}

Let us discuss processing the second type of queries. All strings $s^i$ that can be complement for $t$ belongs to the set $C(t)=\{s^r: r\in\{i_1,\dots,i_{|\cS|}\}, s^r\geq t$ and $s^r<t'$\}. Here we can obtain $t'$ from the string $t$ by replacing the last symbol by the next symbol from the alphabet. Formally, if $t=(t_1,\dots,t_{|t|-1},t_{|t|})$, then $t'=(t_1,\dots,t_{|t|-1},t'_{|t|})$, where the symbol $t'_{|t|}$ succeeds $t_{|t|}$ in the alphabet. We can say, that $C(t)=D(t)$. The query processing consist of three steps.

\textbf{Step 1.} We search for a node $v_{Rt}$ such that $t$ should be in the left sub-tree of $v$ and $t'$ should be in the right sub-tree of $v$. Formally, $\beta(v_{Rt})\leq t\leq \alpha(v_{Rt})$, and $\alpha(v_{Rt})< t'\leq \gamma(v_{Rt})$.
For implementing this idea the procedure $\textsc{IsItTheTarget}(v)$ checks the following condition:
\[\big(\textsc{Compare}(\beta(v),t)\leq 0\mbox{ and }\textsc{Compare}(t,\alpha(v))\leq 0\big)\]\[\mbox{ and  }\]\[\big(\textsc{Compare}(\alpha(v),t')<0\mbox{ and }\textsc{Compare}(t',\gamma(v))\leq 0\big).\]
The procedure $\textsc{SelectChild}(v)$ returns the right child if $t>\alpha(v)$, i.e. \[\textsc{Compare}(t,\alpha(v))>0,\] and returns the left child otherwise.

If we come to the null vertex, then there is no complements of $t$. If we find $v_{Rt}$, then we do next two steps. 
In that two steps, we compute the index of the required string $j_{ans}$ and $jc_{ans}=\#(s^{j_{ans}})$.

\textbf{Step 2.} Let us look to the left sub-tree with $t$. Let us find a node $v_L$ that contains an index $i_L$ of  the minimal string $s^{i_L}\geq t$.
For implementing this idea the procedure $\textsc{IsItTheTarget}(v)$ checks wether the current string is $t$. Formally, it checks $\textsc{Compare}(\alpha(v),t)=0$. Additionally, the procedure saves the node $v_{L}\gets v$ if $\alpha(v)\geq t$, i.e., $\textsc{Compare}(t,\alpha(v))\geq 0$.
The procedure $\textsc{SelectChild}(v)$ works as for searching $t$. It returns the right child if $t>\alpha(v)$, i.e. $\textsc{Compare}(t,\alpha(v))>0$, and returns the left child otherwise. 
In the end, the target node is stored in $v_L$.

Then, we go up from this node. Let us consider a vertex $v$. If it is the right child of $\textsc{Parent}(v)$, then it is bigger than the string from $\textsc{Parent}(v)$ and the left child's sub-tree, so we do nothing. If it is the left child of $\textsc{Parent}(v)$, then it is less than the string from $\textsc{Parent}(v)$ and all strings from the right child's sub-tree, so we update $j_{ans}$ and $jc_{ans}$ by values from the parent node and the right child. Formally, if $(i_p,c_p,j_p,jc_p)$ for the $\textsc{Parent}(v)$ node and $(i_r,c_r,j_r,jc_r)$ for the right child node, then we do the following actions. If $c_p>jc_{ans}$ or ($c_p=jc_{ans}$ and $i_p<j_{ans}$), then $j_{ans}\gets i_p$ and $jc_{ans}\gets c_p$.    If $jc_r>jc_{ans}$ or ($jc_r=jc_{ans}$ and $j_r<j_{ans}$), then $j_{ans}\gets j_r$ and $jc_{ans}\gets jc_r$.
This idea is presented in Algorithm \ref{alg:compl22}

\begin{algorithm}[ht]
    \caption{Obtaining the answer of Step 2 by $v_L$} \label{alg:compl22}
    \begin{algorithmic}
        \State $v \gets v_L$
         \While{$v\neq \textsc{Parent}(v_{Rt})$}
         \State $parent\gets \textsc{Parent}(v)$
         \If{$v=\textsc{LeftChild}(parent)$}
          \State $(i_p,c_p,j_p,jc_p)$ is associated with $parent$ 
          \If{$c_p>jc_{ans}$ or ($c_p=jc_{ans}$ and $i_p<j_{ans}$)}
          \State $j_{ans}\gets i_p$ 
           \State $jc_{ans}\gets c_p$
          \EndIf
          \State $(i_r,c_r,j_r,jc_r)$ is associated with $\textsc{RightChild}(parent)$
          \If{$\textsc{RightChild}(parent)\neq NULL$}
          \If{ ($c_r>jc_{ans}$ or  ($c_r=jc_{ans}$ and $i_r<j_{ans}$))}
          \State $j_{ans}\gets i_r$ 
           \State $jc_{ans}\gets c_r$
          \EndIf
          \EndIf
         \EndIf
         \State $v\gets parent$
         \EndWhile
    \end{algorithmic}
\end{algorithm}

\textbf{Step 3.} Let us look to the right sub-tree with $t'$. Let us find a node $v_R$ that contains an index $i_R$ of  the maximal string $s^{i_R}< t'$. Then, we go up from this node and do the symmetric actions as in Step 2.

Each of these three steps requires $O(\log n)$ running time because each of them observes nodes of a single branch. The complexity of the quantum algorithm is presented in Theorem \ref{th:comlete-quantum}. Before presenting the theorem, let us present a lemma that helps us to prove quantum and classical lower bounds.
\begin{lemma}\label{lm:unstruct-search}
Auto-complete problem at least as hard as unstructured search $1$ among $O(L)$ bits.
\end{lemma}
\begin{proof}
Assume that the alphabet is binary. For any other case we just consider two letter from the alphabet.
Assume that all strings from queries of the first type have length $k$.

Let $m = \lfloor(n-1)/2\rfloor$.
Let us consider the following case. We have $m$ queries of the first type of the next form: $s^i = (0,0,s^i_3,\dots,s^i_k)$, Here $s^i_j=0$ for all $i\in\{1,\dots,m\}$ and $j\in\{3,\dots,k\}$ except one bit. We have two cases.

The first case is the following. There is only one pair $(r,w)$ such that $r\in\{1,\dots,m\}$, $w\in\{3,\dots,k\}$, and $s^r_w=1$. In the second case there is no such a pair.

Next  $m$ queries of the first type of the next for: $(0,1,0,\dots,0)$.

The if $n-1$ is odd, then we add a query of the first type of the next for: $(1,\dots,1)$.

The last query of the second type of the form $t=(0)$.

If we have the first case, then $\#(0,1,0,\dots,0)=m$, $\#(0,0,0,\dots,0)=m-1$ and the answer is  $(0,1,0,\dots,0)$. If we have the second case, then $\#(0,1,0,\dots,0)=m$, $\#(0,0,0,\dots,0)=m$ and $(0,0,0,\dots,0)$ has a smaller index. Therefore, the answer is  $(0,0,0,\dots,0)$.

Hence, the answer for the input is at least as hard as distinguish between these two cases. At the same time, it requires search $1$ among $O((k-3)m)=O(L)$ bits.
\Endproof
\end{proof}
Based on the presented lemma and the discussed algorithm, we can resent quantum lower and upper bounds for the problem in the next theorem. 

\textbf{Theorem \ref{th:comlete-quantum}.}\textit{
The quantum algorithm with a noisy Self-Balanced Search tree for Auto-Complete Problem has $O(\sqrt{nL}\log n)$ running time and error probability $1/3$, where $L$ is the sum of lengths of all queries. Additionally, the lower bound for quantum running time is $\Omega(\sqrt{L})$.
}
\begin{proof}
Let us start with the upper bound. Let us consider processing of the first type of queries. Here we add a string $s$. Firstly, we find the target node with $O(\sqrt{|s|}\log n)$ running time and $1/n^2$ error probability according to Corollary \ref{cr:bst-compl}. Then, we consider at most $O(\log n)$ ancestors for updating with $O(\log n)$ running time and no error. So, processing the first type of queries works with $O(\sqrt{|s|}\log n)$ running time and  $1/n^2$ error probability.

 Let us consider processing of the second type of queries. Here we search for complement for a string $t$. Searching nodes $v_{Rt}$, $v_L$ and $v_R$ works with $O(\sqrt{|t|}\log n)$ running time and $1/n^2$ error probability according to Corollary \ref{cr:bst-compl}. Then,  we consider at most $O(\log n)$ ancestors for updating with $O(\log n)$ running time and no error.  So, processing the second type of queries works with $O(\sqrt{|t|}\log n)$ running time and  $3/n^2$ error probability.
 
 Let $m_1,\dots,m_n$ be lengths of strings from queries. So, the total complexity is 
 \[O(\sqrt{m_1}\log n+\dots+\sqrt{m_n}\log n)=O((\sqrt{m_1}+\dots+\sqrt{m_n})\log n)=\]\[
 O(\sqrt{n(m_1+\dots+m_n)}\log n=O(\sqrt{nL}\log n)\]
 The last two equalities due to Cauchy--Bunyakovsky--Schwarz inequality and $L=m_1+\dots+m_n$.
 
 Error probability of processing a query is at most $3/n^2$. Therefore, error probability of processing $n$ queries is at most $0.1$ due to all error events are independent.
 
 Let us discuss the lower bound. It is known \cite{bbbv1997} that the quantum running time for the unstructured search among $O(L)$ variables is $\Omega(\sqrt{L})$. So, due to Lemma \ref{lm:unstruct-search} we obtain the required lower bound.
\Endproof \end{proof}

Let us consider the classical (deterministic or randomize) case. If we use the same Self-balanced search tree, then the running time is $O(L\log n)$. At the same time, if we use the Trie (prefix tree) data structure 
\cite{knuth73}, then the complexity is $O(L)$.

We store all string of $\cS$ in the trie. For each terminal node we store ``frequency'' of the corresponding string. For each node $v$ (even non-terminal) that corresponds to a string $u$ as a path from the root to $v$, additionally to the regular information we store the index of required complement for $t$ and its frequency. When we process the first type of query, we update the frequency in the terminal node and update additional information in all ancestor nodes because they store all possible prefixes of the string. For processing the query of the second type, we just find the required node and take the answer in the additional information of the node.

We can show that it also the lower bound for classical case.

\textbf{Lemma \ref{lm:comlete-classical}.} \textit{
The classical running time for Auto-Complete Problem  is $\Theta(L)$, where $L$ is the sum of lengths of all queries. }
\begin{proof}
Let us start with the upper bound. Similarly to the proof of Theorem \ref{th:comlete-quantum}, we can show that the running time of processing a query of both types is $O(|s|)$ or $O(|t|)$ depends on the type of a query.
 Let $m_1,\dots,m_n$ be lengths of strings from queries. So, the total complexity is $O(m_1+\dots+m_n)=O(L)$.
 
 Let us discuss the lower bound. It is known \cite{bbbv1997} that the classical running time for the unstructured search among $O(L)$ variables is $\Omega(L)$. So, due to Lemma \ref{lm:unstruct-search} we obtain the required lower bound.
\Endproof \end{proof}

If $O(n)$ strings  have length at least $\Omega((\log_2 n)^2)$, then we obtain a quantum speed-up.

\end{document}